\documentclass[11pt, letterpaper]{article}

\usepackage[T1]{fontenc}
\usepackage{lmodern}
\usepackage{amsmath, amssymb, amsthm, dsfont, mathtools}
\usepackage{thmtools}
\usepackage{thm-restate}
\usepackage{xspace}
\usepackage{graphicx}
\usepackage{tikz}
\usepackage{nicefrac}
\usepackage{aligned-overset}
\usepackage[dvipsnames]{xcolor}
\definecolor{darkgreen}{RGB}{0,128,43}

\usetikzlibrary{calc, backgrounds, decorations.markings, arrows.meta}

\usepackage{enumitem}
\newcommand{\aff}[1]{\textcolor{black!50}{#1}}

\title{A tight lower bound for malicious online bipartite matching \\ with limited recourse budget\thanks{\textbf{Funding:} This work is a result of research conducted within project number 2023/51/B/ST6/02833 financed by the National Science Centre, Poland (Bartłomiej Bosek, Paweł Putra and Anna Zych-Pawlewicz).\\
Júlia Baligács was funded by the European Union through the European Research Council under the project BOBR (grant agreement No.~948057) during employment in Warsaw and under the project CCOO (grant agreement No.~101165139) during employment in Oxford. Views and opinions expressed are however those of the authors only and do not necessarily reflect those of the European Union or the European Research Council. Neither the European Union nor the granting authority can be held responsible for them.\\
Marek Sokołowski was supported by the Deutsche Forschungsgemeinschaft (DFG, German Research Foundation) grant number 559177164.}}

\author{
  Júlia Baligács\\{\small\aff{University of Oxford}}\\
    \href{mailto:jbaligacs@gmail.com}{\small jbaligacs@gmail.com}\bigskip
    \and
  Bart{\l}omiej Bosek\\{
  \small\aff{Jagiellonian University, Krak\'ow}}\\
    \href{mailto:bartlomiej.bosek@uj.edu.pl}{\small bartlomiej.bosek@uj.edu.pl}\bigskip
    \and 
  Paweł Putra\\{\small\aff{University of Warsaw}}\\
    \href{mailto:p.putra@uw.edu.pl}{\small p.putra@uw.edu.pl}
  \and
  Marek Sokołowski\\{\small\aff{Max Planck Institute for Informatics, SIC, Saarbrücken}
  }\\
    \href{mailto:msokolow@mpi-inf.mpg.de}{\small msokolow@mpi-inf.mpg.de}
  \and
  Anna Zych-Pawlewicz\\{\small\aff{University of Warsaw}}\\
    \href{mailto:anka@mimuw.edu.pl}{\small anka@mimuw.edu.pl}
  }

 \date{}

\usepackage[margin=1in]{geometry}
\usepackage[sort,numbers]{natbib}

\DeclareMathOperator{\girth}{girth}
\DeclareMathOperator{\dist}{dist}

\usepackage{hyperref}
\usepackage[nameinlink]{cleveref}

\crefformat{lemma}{#2Lemma~#1#3}
\crefformat{theorem}{#2Theorem~#1#3}
\crefformat{proposition}{#2Proposition~#1#3}
\crefformat{remark}{#2Remark~#1#3}
\crefformat{equation}{#2(#1)#3}
\crefformat{corollary}{#2Corollary~#1#3}
\crefformat{example}{#2Example~#1#3}
\crefformat{definition}{#2Definition~#1#3}
\crefformat{figure}{#2Figure~#1#3}
\crefformat{observation}{#2Observation~#1#3}
\crefformat{section}{#2Section~#1#3}
\crefformat{question}{#2Question~#1#3}

\definecolor{LinkViolet}{hsb}{0.37,1,0.5}
\hypersetup{
  colorlinks=true,
  linkcolor=Violet,
  citecolor=Violet,
  urlcolor=green!35!black
}

\theoremstyle{definition}
\newtheorem{definition}{Definition}

\theoremstyle{plain}

\newtheorem{lemma}[definition]{Lemma}

\newtheorem{observation}[definition]{Observation}

\newcommand{\N}{\mathbb{N}}
\newcommand{\R}{\mathbb{R}}
\newcommand{\bigO}{\mathcal{O}}

\newcommand{\sap}{\ensuremath{\mathsf{SAP}}\xspace}

\newcommand{\obm}{\textsc{Malicious Online Bipartite Matching with Recourse}\xspace}
\newcommand{\obmr}{\textsc{Online Bipartite Matching with Recourse}\xspace}

\newcommand{\bra}[1]{\ensuremath{({#1})}}
\newcommand{\fO}[1]{\ensuremath{\mathcal{O}\bra{#1}}}
\newcommand{\fTheta}[1]{\ensuremath{\Theta\bra{#1}}}
\newcommand{\tuple}[1]{\ensuremath{\langle\hspace{0.5mm}{#1}\hspace{0.5mm}\rangle}}

\newcommand{\size}[1]{\ensuremath{{|}{#1}{|}}}

\tikzset{
  midarrow/.style={
    postaction={decorate},
    decoration={markings, mark=at position .7 with {\arrow[#1]{Stealth}}}
  }
}

\begin{document}
\renewcommand\footnotemark{}
\maketitle
\begin{abstract}
We study one-sided online bipartite matching with recourse.
In this setting, one side of a bipartite graph is known in advance, while vertices on the other side arrive online together with their incident edges.
After each arrival, the algorithm must maintain a maximum-cardinality matching while minimizing the total number of reallocations, also known as the recourse budget.
Despite extensive work, the exact recourse complexity of the problem remains unsettled: the best lower bound is $\Omega(n \log n)$, whereas the best upper bound is $\mathcal{O}(n \log^2 n)$, where $n$ denotes the number of online vertices.
Tight upper bounds of $\mathcal{O}(n \log n)$ are known only for restricted graph classes, such as forests. 

The best known upper bounds are attained by a very simple and natural algorithm \sap, which after each arrival applies a shortest augmenting path, and it is conjectured to be optimal. 
All known upper bound analyses of this algorithm do not depend on the particular maximum matching maintained by the algorithm.
Consequently, they also apply to a more difficult problem, which we call the malicious matching setting: after each arrival, the maintained matching is replaced by a worst-case maximum matching for the next step.
This led to the conjecture that the malicious setting still admits an $\mathcal{O}(n \log n)$ recourse bound, in line with the conjectured optimal complexity of the original model.

Our main result is an $\Omega(n \log^2 n)$ lower bound for the malicious matching setting, thus disproving the conjecture.
Together with the previous upper bound, this settles the asymptotic recourse complexity of the malicious variant of the problem.
We complement our lower bound with an upper bound of $\mathcal{O}(n \log n)$ for expander graphs.
\end{abstract}

\section{Introduction}
One of the most fundamental problems in computer science with a long history of research devoted to it is the bipartite matching problem. In 1973, Hopcroft and Karp~\cite{HopcroftK73} presented an elegant deterministic algorithm that computes
a maximum
matching in a bipartite graph in time $\fO{|E| \sqrt{|V|}}$, where~$|V|$
is the number of vertices and $|E|$ is the number of edges in the
graph. This algorithm was a milestone, and its running time was hard to beat.
Improvements were obtained for dense graphs~\cite{SankMuchaFOCS04} and for sparse graphs~\cite{MadryFOCS13}, but in the general case, despite countless efforts, no faster algorithm was known for nearly five decades.
In a recent breakthrough~\cite{linearMatching}, a randomized algorithm with running time~$\fO{|E|^{1+o(1)}}$ was presented for the bipartite matching problem.

Since the early work of Hopcroft and Karp, due to its importance, the bipartite matching problem has been studied in a wide variety of dynamic models, where the bipartite graph is updated over time and the goal is to efficiently maintain a good-quality matching at all times~\cite{DBLP:conf/focs/AbboudW14, BernsteinS15, 10.5555/2884435.2884524,Sankowski07,Solomon,GuptaP13,bernstein_et_al:LIPIcs.FSTTCS.2020.11,10.1007/s10878-020-00641-w,10.1145/2897518.2897568,doi:10.1137/1.9781611974782.30,Bernstein2016FasterFD}. 
In this paper, we study the one-sided online model with limited recourse budget introduced by Grove, Kao, Krishnan, and Vitter in 1995~\cite{GroveKKV95} (see~\cref{sec:preliminaries} for a formal definition). 
This model captures a scenario in which clients arrive online and request service from a given set of servers. 
Upon arrival, the client reveals the set of servers that can provide the service desired by the client.
Our objective is to serve as many clients as possible, that is, to maintain a maximum-cardinality matching after each arrival.
Thus, whenever the newly arrived client can be matched in a way that increases the size of the current matching, we assign a server to the client, possibly reallocating previously matched clients.
The goal is to minimize the total reallocation cost over all arrivals, referred to in the literature as the \emph{recourse budget}.

This model of online bipartite matching has a broad range of applications across computer science, including streaming content delivery, web hosting, remote data storage, job scheduling, and hashing (see~\cite{ChaudhuriDKL09} for a detailed discussion).
In these applications, clients may represent, for example, jobs that must be assigned to processing machines, or objects that must be placed in cells of a hash table.
In such settings, dropping clients from service is undesirable; instead, we prefer to reallocate them so as to serve as many clients as possible.
The cost of not serving a client is typically much higher than the cost of reallocation: for instance, copying an object in a hash table is relatively cheap and much more desirable than dropping it from the table entirely.
Nevertheless, it remains important to minimize the total reallocation cost. 
Another motivation for studying this model is its potential to give rise to simple yet efficient offline algorithms, as it was the case in~\cite{6979023,BernsteinHR19}.
This setting contrasts with the classical one-sided online model introduced by Karp, Vazirani, and Vazirani~\cite{KarpVV90}, in which matching decisions are irrevocable.
There, upon the arrival of a client, the algorithm must decide whether to match the client and, if so, to which server, and this decision cannot be changed later.
The objective in that model is to maximize the number of matched clients subject to this irrevocability constraint.
For this model, there are pessimistic bounds on the number of clients that can be served~\cite{KarpVV90}, which however do not apply to our recourse setting.

The first notable result on the one-sided online model with limited recourse was obtained in~\cite{GroveKKV95}, where the authors considered the case in which each client has degree at most two and they proved a tight bound of $\fTheta{n \log n}$ on the recourse budget. Here, $n$ is the total number of inserted clients, which without loss of generality can be assumed to be also the number of servers (this assumption is very common in the related work that we describe below).
Next, in~\cite{ChaudhuriDKL09}, a bound of $\fO{n \log n}$ was established for further restricted scenarios, for instance when the clients arrive in a random order.
The authors studied the most natural greedy algorithm, which extends the matching using a shortest augmenting path upon each arrival, and is referred to as the Shortest Augmenting Path algorithm (\sap).
This led them to conjecture that the \sap algorithm achieves a recourse budget of $\fO{n \log n}$ also for general bipartite graphs under adversarial arrivals.
Nevertheless, at that time it was not even clear whether, in this setting, any recourse budget better than the trivial $\fO{n^2}$ could be achieved, and this was left as the main open question.

The first nontrivial upper bound on the recourse budget, namely $\fO{n \sqrt{n}}$, was proved in~\cite{6979023}, using an algorithm more sophisticated than \sap, but also elegant and simple.
Finally, in the breakthrough paper by Bernstein, Holm, and Rotenberg~\cite{BernsteinHR19}, the \sap algorithm was shown to achieve a total recourse budget of $\fO{n \log^2 n}$, which is within a factor of $\fO{\log n}$ of the known lower bound of~$\Omega(n \log n)$. 
Closing this gap was left as a main open problem in the area.
Shortly after, an upper bound of $\fO{n \log n}$ was established for forests~\cite{BosekSAPtrees22}, and the authors conjectured that forests are the hardest case for \sap. The result of Bernstein, Holm and Rotenberg~\cite{BernsteinHR19} was later generalized to the matroid setting~\cite{DBLP:conf/soda/BuchbinderGHKS24}, still with the same upper bound on the recourse budget.

The \sap algorithm, although it is the most natural strategy for the problem, has its limitations.
In particular, the analyses of \sap in~\cite{BernsteinHR19,BosekSAPtrees18,BosekSAPtrees22} note that shortest augmenting paths can change the matching quite erratically.
Thus, these analyses do not rely on specific properties of the maintained matching, but rather on the structural properties of the portion of the graph that has already been revealed.
As a consequence, a careful inspection of the proofs in~\cite{BernsteinHR19,BosekSAPtrees18,BosekSAPtrees22} shows that all these results extend to a more difficult scenario, in which after the arrival of each client the maintained matching is changed into the worst possible matching for the \sap algorithm.
This setting was introduced in~\cite{BosekSAPtrees18}, where it was called the adversarial dynamic augmenting path setting.
We refer to it as the \emph{malicious matching} setting.

It is easy to see that, in the malicious matching setting, \sap is the best possible strategy: there is nothing to gain from choosing an augmenting path other than a shortest one.
In light of the above, it is natural to ask whether the upper bound of $\fO{n \log^2 n}$ can be improved for general bipartite graphs in the malicious matching setting.
In fact, the authors of \cite{BosekSAPtrees18} conjectured an upper bound of $\bigO(n \log n)$ for this model, in line with the conjectured bound for the non-malicious setting.
In this work, we answer this question in the negative, thereby disproving the conjecture.

\paragraph*{Our results.}
Our main result is the following lower bound for \obm.

\begin{restatable}{theorem}{thmlowerbound}
\label{thm:lowerbound}
For infinitely many $n\in \N$, there exists an instance for \obm on $n$ servers and $n$ clients on which any algorithm incurs total recourse budget at least $\Omega(n \log^2 n)$.
\end{restatable}

Since the upper bound proved in~\cite{BernsteinHR19} also applies to the malicious matching setting, the lower bound of \cref{thm:lowerbound} is tight and thus settles this problem.
Moreover, since an upper bound of~$\bigO(n \log n)$ is known for forests~\cite{BosekSAPtrees22}, this separates forests from the general setting.
It also disproves the conjecture made in~\cite{BosekSAPtrees22} that forests are the hardest case for \sap, and the conjectured recourse bound in \cite{BosekSAPtrees18} for the malicious setting.

We remark that the graphs in our construction are subcubic.
Hence, while the case where clients have degree at most two admits a bound of $\bigO(n \log n)$~\cite{GroveKKV95}, any larger degree bound does not simplify the problem.

Another way of viewing the result of \cref{thm:lowerbound} is as follows.
If one aims to prove the conjectured bound of $\bigO(n \log n)$ for online bipartite matching with recourse in the non-malicious setting, then the techniques currently available in the literature are already pushed to their limits, and substantially new ideas are needed.

We complement our lower bound with the following upper bound for expander graphs, where the precise definition of an expander needed here is deferred to \cref{sec:expander}.

\begin{restatable}{theorem}{expanderupperbound}
    \label{thm:expanders-formal}
    \label{thm:expanders}
    Let $h > 0$ and $d \geq 1$.
    On every instance of \obm{} on $n$ servers and $n$ clients
    containing a~$d$-regular $h$-edge-expanding spanning subgraph,~\sap{} achieves total recourse budget $\bigO_{h, d}(n \log n)$.
\end{restatable}

We find this result interesting for two reasons. First, supergraphs of expanders are well-connected and can be dense,
in contrast to other classes for which a budget $\bigO(n \log n)$ has been proven, such as forests and graphs where clients have degree at most two. Our result, together with previous work, shows that the two extreme cases admit the stronger bound, whereas the intermediate case is genuinely harder.
Second, the analysis giving $\bigO(n \log^2 n)$ recourse budget~\cite{BernsteinHR19} crucially relies on a~certain bound on the number of long augmenting paths \cite[Lemma 6]{BernsteinHR19}.
During our work, one of our first steps was a construction of an expander graph that proves this bound to be tight.
Indeed, our final construction is an expander-like graph with some degree-one vertices attached.
Thus, \cref{thm:expanders} can be viewed as a key step towards understanding the structure needed for the proof of \cref{thm:lowerbound}.

\paragraph{Outline and overview of techniques.}
In \cref{sec:preliminaries}, we present preliminaries and give a formal description of the problem studied in this paper.

In \cref{sec:lowerbound}, we turn to the proof of our main result, the lower bound of $\Omega(n \log^2 n)$ for \obm.
The proof consists of two key ingredients.
In \cref{sec:matching_instance}, we first show how to turn a graph~$G$ satisfying certain key properties concerning its girth (i.e., the length of a shortest cycle) into a difficult instance of the problem.
The girth plays an important role for us since we make use of the fact that a graph of girth $g$ looks locally like a tree within any ball of radius less than $(g-1)/2$.
A second key ingredient is that, once a client of degree one is revealed and matched to its neighbor, no augmenting path can later pass through this client or through its neighboring server.
In the previous literature, such vertices were called \emph{dead} and can be thought of as being deleted from the graph.
In \cref{sec:construct_G}, we then carefully construct a graph that initially has girth $\Theta(\log n)$ and has the following property: if we remove edges in a certain order, its girth gradually and sufficiently increases, eventually becoming~$\Theta(n)$.
We note that this construction can be viewed as a purely graph-theoretic result and might be of independent interest.
At the same time, it provides exactly the missing ingredient needed for our lower bound construction, and applying our results from \cref{sec:matching_instance} to this graph yields our main result, \cref{thm:lowerbound}.

Finally, in \cref{sec:expander}, we prove our upper bound for expander graphs, \cref{thm:expanders}.
The key property of expanders that we use is the following: the neighborhood of any not-too-large set is larger than the set itself by a constant factor.
Using this, we are able to prove that the number of servers reachable from a client by an alternating path of length $i$ increases exponentially in $i$.
From this, we conclude that, at every time step, there exists an augmenting path of length at most~$\bigO(\log n)$.

\paragraph{Comment on AI usage.}
During the proof-discovery process of \cref{thm:lowerbound}, we used the large language model GPT-5.5 Pro.
Before using this tool, we had developed a preliminary, weaker version of the lower bound.
The tool was then used in an iterative discussion to refine the lower bound, which helped find the final construction presented in this paper.

\section{Preliminaries}
\label{sec:preliminaries}

First, we collect the graph-theoretic notation used throughout the paper and then formalize the model that we study.

\paragraph{Basic graph notation.}
For a graph $G=(V,E)$, we write $V(G)$ for its vertex set and $E(G)$ for its edge set.
A \emph{matching} in $G$ is a set of pairwise non-incident edges.
For a set of vertices $X \subseteq V(G)$, we denote by $N_G(X)$ the open neighborhood of vertices from $X$, while by $G[X]$ we denote the subgraph of $G$ induced by $X$; that is,
$
G[X]=(X,\{\{u,v\}\in E(G): u,v\in X\}).
$
The girth of $G$, denoted by $\girth(G)$, is the length of a shortest cycle in $G$ (and is $\infty$ if $G$ is acyclic).
A Hamiltonian cycle in $G$ is a simple cycle that contains all vertices of $G$, and we say that $G$ is \emph{Hamiltonian} if it contains such a cycle.
Throughout the paper, $\log$ denotes the logarithm to base~$2$.

\paragraph{The model.}

Next, we formally introduce the setting and notation that we work with. We start by defining the classical \obmr problem, and then proceed to the variant with malicious matchings.

An instance of the \obmr problem is 
a bipartite graph~$B = \tuple{S, C, E_B}$,  where $S$ is the set of \emph{servers}, $C$ is the set of \emph{clients}, and $E_B$ is the edge set.
The set of servers $S$ is known to the algorithm from the beginning, whereas the clients are revealed online.
More precisely, the clients are revealed according to a \emph{presentation order}~$c_1,\ldots,c_n$, where $C=\{c_1,\ldots,c_n\}$.
For $t\in\{1,\ldots,n\}$, let $C_t:=\{c_1,\ldots,c_t\}$ and $B_t:=B[C_t\sqcup S]$.
At time~$t$, the client~$c_t$ and all its incident edges to $S$ are revealed.
The algorithm must then construct a maximum (i.e., maximum-cardinality) matching $M_t$ in $B_t$, knowing only $B_t$ and the previously constructed matchings $M_1,\ldots,M_{t-1}$.
We set $M_0=\emptyset$.
The recourse budget paid at time $t$ is $r_t:=\size{M_t\oplus M_{t-1}}$ (where~$\oplus$ denotes symmetric difference).
The \emph{total recourse budget} of the algorithm is $r:=\sum_{t=1}^n r_t$.
The objective is to minimize $r$.

In \obm, at each time $t$, before the algorithm computes matching~$M_t$, it is additionally presented  an adversarially chosen maximum matching~$M'_{t-1}$ of $B_{t-1}$. 
The algorithm then has to work with this matching $M'_{t-1}$ instead of its maintained matching~$M_{t-1}$, i.e., it has to compute a maximum matching $M_t$ of $B_t$ knowing $M'_{t-1}$, and the recourse budget at time~$t$ is then  $r_t':=\size{M'_{t-1} \oplus M_t}$.
The goal is again to minimize the total recourse budget $r':=\sum_{t=1}^n r_t'$.

\paragraph{Augmenting paths.}
Let $m^*_t$ denote the size of a maximum matching in $B_t$. 
Fix an arbitrary maximum matching $M$ in $B_{t-1}$.
An \emph{alternating path} in $B_t$ with respect to $M$ is a path whose edges alternate between edges outside $M$ and edges inside $M$.
For our purposes, an \emph{augmenting path} in~$B_t$ with respect to~$M$ is an alternating path that starts at the newly revealed client $c_t$ and ends at a server $s\in S$ that is unmatched by $M$.
By Berge's lemma, $m^*_t=m^*_{t-1}+1$ if and only if an augmenting path exists in $B_t$, and since it does not exist in $B_{t-1}$, this augmenting path starts at~$c_t$ as required.

Note that, if $m^*_t=m^*_{t-1}$, then the recourse $r_t$ of any reasonable algorithm is~$0$, as there is no need to change the matching since it is already maximum. 
In particular, we assume that in that case the shortest augmenting path with respect to any maximum matching has length~$0$.
Observe that, in the interesting case 
where $m^*_{t-1}< m^*_t$, i.e., the algorithm is forced to change the matching, the symmetric difference $M'_{t-1} \oplus M_t$ is a collection of alternating paths and cycles, containing an augmenting path in $B_t$ with respect to $M'_{t-1}$. 
Thus, the matching $M_t$ minimizing $\size{M'_{t-1} \oplus M_t}$ is such that $M'_{t-1} \oplus M_t$ is precisely a shortest possible augmenting path. The algorithm which applies the shortest augmenting path after each client's arrival is referred to as the Shortest Augmenting Path algorithm (\sap). 
Based on the above, in the malicious matching setting, the \sap algorithm is optimal with respect to the recourse budget. 

One final remark is that one can assume without loss of generality that $|S|=|C|=n$ and that there is a perfect matching in $B_n$ (i.e., this is the most difficult case). This was observed by many papers~\cite{ChaudhuriDKL09,6979023,BosekSAPtrees18,BosekSAPtrees22} in the non-malicious setting, and the arguments carry over to the malicious setting as well. Nevertheless, since our results do not rely on this observation, we only mention it here for completeness.

\section{A tight lower bound for the malicious setting}
\label{sec:lowerbound}
\newcommand{\hamcycle}{H}

In this section, we prove our main result, \cref{thm:lowerbound}.
We begin by restating it.

\thmlowerbound*

We prove the theorem by separating the argument into two independent ingredients.
We first describe the two ingredients and show that they together imply \cref{thm:lowerbound}.
Their proofs are then given in the following two subsections.

The first ingredient converts a Hamiltonian graph $G$ with chord set $D$ into an instance of \obm. 
The construction is on a bipartite graph derived from $G$, but the resulting process can be interpreted as the following game on $G$: at each step, the adversary selects a chord $e\in D$, forces the algorithm to incur cost at least the girth of the current graph (up to an additive constant), and then deletes $e$.

This interpretation is useful because deleting edges cannot decrease the girth. 
We will apply this process to a graph with initial girth $\Theta(\log n)$ for which the girth increases gradually as chords are deleted, eventually becoming $\Theta(n)$.

In the following, the notation $G - E'$ for a graph $G$ and a subset of its edges $E'$ denotes the graph obtained by deleting all edges in $E'$ from $G$.

\begin{restatable}{proposition}{propmatchinginstance}
\label{prop:matching_instance}
Let $G=(V,E)$ be a graph containing a Hamiltonian cycle $\hamcycle$.
Let $D=E\setminus \hamcycle$ be the set of chords, and let $d_1,\dots, d_k$ be an ordering of $D$.
Then there exists an instance of \obm with $|E|$ clients and $|E|$ servers, on which any algorithm incurs total recourse at least
\begin{equation*}
    \sum_{j=1}^{k} \girth(G - \{d_1,\dots, d_{j-1}\}) -5k.
\end{equation*}
\end{restatable}

The next proposition provides the graph-theoretic ingredient and constructs the graph on which we will apply \cref{prop:matching_instance}.
In the following, we identify subgraphs (and Hamiltonian cycles) with their set of edges, and vice versa, we identify a set of edges with its edge-induced subgraph.

\begin{restatable}{proposition}{propconstructiong}
\label{prop:construction_G}
There exist constants $\alpha, \beta>0$ such that, for every integer $L\geq 2$, there exists a graph~$G=(V,E)$ on $|V|=2^L$ vertices with the following properties.
$G$ contains a Hamiltonian cycle $\hamcycle$. Its chord set $D = E \setminus \hamcycle$ is a matching and admits a partition $D = F_1 \cup \dots \cup F_{L-1}$ such that, for every $i\in \{1,\dots, L-1\}$,
\begin{enumerate}[label=(\roman*)]
\item $|F_i|\geq \alpha \cdot 2^i$,
\label{cond:sizeF}
\item $\girth \left( \hamcycle \cup \bigcup_{j=1}^i F_j \right) \geq \beta \cdot i \cdot 2^{L-i}.$
\label{cond:girth}
\end{enumerate}
\end{restatable}

Before proving the two propositions, we show that together they imply \cref{thm:lowerbound}.

\begin{proof}[Proof of \cref{thm:lowerbound} using Propositions~\ref{prop:matching_instance}~and~\ref{prop:construction_G}]
Let $L \ge 2$ be arbitrary, and let ${G=(V,E)}$ be the graph given by \cref{prop:construction_G}.
Let $\hamcycle$ be its Hamiltonian cycle, and let ${D = F_1 \cup \dots \cup F_{L-1}}$ be the corresponding partition of its chord set.
Set $N := |V| = 2^L$ and $k := |D|$.

First, note that $k \le N/2$, since $D$ is a matching.
Moreover, condition \ref{cond:sizeF} of \cref{prop:construction_G} implies that
$k \ge |F_{L-1}| \ge \alpha 2^{L-1} = N\alpha/2$.
Thus $|E|=k+N = \Theta(N)$.

We order the chords by first listing all chords in $F_{L-1}$, then all chords in $F_{L-2}$, continuing in decreasing order of the index, and finally all chords in $F_1$.
The order within each set $F_i$ is arbitrary.
Consider a chord $d_j \in F_i$.
Since all chords in $F_{i+1},\dots,F_{L-1}$ appear before $d_j$ in the ordering, we have \begin{equation}
\label{eq:girth_Fj}
\girth (G - \{d_1, \dots, d_{j-1}\}) \geq \girth \Bigl( G - {\textstyle \bigcup_{l=i+1}^{L-1}} F_l \Bigl) = \girth \Bigl(\hamcycle \cup {\textstyle \bigcup_{l=1}^i} F_l \Bigl) 
\overset{\text{Prop.}\ref{prop:construction_G}\ref{cond:girth}}{\geq} \beta \cdot i \cdot 2^{L-i}.
\end{equation}

By \cref{prop:matching_instance}, there exists an instance of \obm of size $|E|=k+N$ on which every algorithm incurs total recourse at least the following amount.
\begin{align*}
r &\geq \sum_{j=1}^k \girth(G - \{d_1,\dots, d_{j-1}\}) -5k
\overset{\eqref{eq:girth_Fj}}{\geq}  
\sum_{i=1}^{L-1} |F_i| \cdot \beta \cdot i \cdot 2^{L-i} -5k\\
\overset{\text{Prop.~}\ref{prop:construction_G}\ref{cond:sizeF}}&{\geq}
\sum_{i=1}^{L-1} \alpha \cdot 2^i \cdot \beta \cdot i \cdot 2^{L-i} -5k
= \alpha \beta \cdot 2^{L} \cdot \sum_{i=1}^{L-1} i -5k
= \alpha \beta \cdot 2^L \cdot \frac{L(L-1)}{2} -5k.
\end{align*}
Since $L=\log N$ and $k=\Theta(N)$, this yields $r = \Omega(2^L L^2) = \Omega(N \log^2 N)$.
Since the number of clients $n$ fulfills $n = \Theta(N)$, this is also $\Omega(n \log^2 n)$. As $L\geq 2$ was arbitrary, this gives a family of instances of unbounded size $n$ with total recourse $\Omega(n \log^2 n)$.
\end{proof}

\subsection{Constructing a lower bound instance from high-girth graphs}
\label{sec:matching_instance}

In this subsection, we prove \cref{prop:matching_instance}.
Throughout the subsection, let $G=(V,E)$ be a graph with Hamiltonian cycle $\hamcycle$ and let $D:=E\setminus \hamcycle$ be its set of chords with the given ordering $d_1,\dots, d_k$.

We first construct a bipartite graph from $G$ and specify the order in which the clients are revealed.
To define an instance of \obm, it then only remains to describe the adversary's matching in each step.

\begin{figure}
\centering
\begin{tikzpicture}[scale=1.2, every node/.style={font=\small, circle, draw, fill, inner sep=0pt, minimum size=6pt}, baseline={(current bounding box.center)}]
\node (a) at (0,0) {}; 
\node (b) at (1,0) {}; 
\node (c) at (1,1) {}; 
\node (d) at (0,1) {};
\node at (-0.3,0) [draw=none, fill=none] {$a$};
\node at (-0.3,1) [draw=none, fill=none] {$d$};
\node at (1.3,0) [draw=none, fill=none] {$b$};
\node at (1.3,1) [draw=none, fill=none] {$c$};
\draw[thick] (a) to (b) to (c) to (d) to (a) to (c);
\end{tikzpicture}
\hspace{3cm}
\begin{tikzpicture}[scale=0.7, every node/.style={font=\small, fill, circle, draw, inner sep=0pt, minimum size=6pt}, baseline={(current bounding box.center)}]
\node (a) at (0,0) {}; 
\node (b) at (0,-1) {}; 
\node (c) at (0,-2) {}; 
\node (d) at (0,-3) {}; 
\node at (-0.5,0) [draw=none, fill=none, rectangle] {$c_a$};
\node at (-0.5,-1) [draw=none, fill=none] {$c_b$};
\node at (-0.5,-2) [draw=none, fill=none] {$c_c$};
\node at (-0.5,-3) [draw=none, fill=none] {$c_d$};
\node (ab) at (2.5,0) {};
\node (bc) at (2.5,-1) {};
\node (cd) at (2.5,-2) {};
\node (ad) at (2.5,-3) {};
\node at (3,0) [draw=none, fill=none, anchor=west, rectangle] {$s_{\{a,b\}}$};
\node at (3,-1) [draw=none, fill=none, anchor=west] {$s_{\{b,c\}}$};
\node at (3,-2) [draw=none, fill=none, anchor=west] {$s_{\{c,d\}}$};
\node at (3,-3) [draw=none, fill=none, anchor=west] {$s_{\{a,d\}}$};
\draw[thick] (a) to (ab) to (b) to (bc) to (c) to (cd) to (d) to (ad) to (a);
\node (ac) at (2.5,-4) {};
\node at (3,-4) [draw=none, fill=none, anchor=west, rectangle] {$s_{\{a,c\}}$};
\draw[thick] (a) to (ac) to (c);
\node (chordclient) at (0,-4) [blue] {};
\node at (-0.8,-4) [draw=none, fill=none, blue, rectangle] {$c_{\{a,c\}}$};
\draw[thick, blue] (chordclient) to (ac);
\end{tikzpicture}
\caption{The figure illustrates a graph $G$ on the left and the resulting constructed instance for \obm on the right.
The client $c_{\{a,c\}}$ (depicted in blue) is in this case the only chord-client.
The order in which the clients are revealed is from top to bottom.
The graph induced by all vertices except $c_{\{a,c\}}$ is the incidence graph of $G$.}
\label{fig:incidence_graph}
\end{figure}
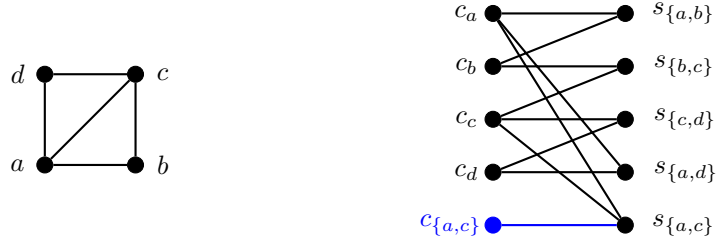

\begin{definition}
\label{def:matching_instance}
The instance resulting from $G$ is defined as follows (cf.~\cref{fig:incidence_graph}).
We define the set of servers to be a copy of $E$, denoted by $S=\{s_e: e\in E\}$.
The set of clients consists of two types.
First, for every vertex $v\in V$, we introduce a \emph{vertex-client} $c_v$, which is adjacent to $s_e$ for every edge $e$ incident to $v$.
The graph defined so far is the \emph{incidence graph} of $G$.
Second, for every chord $e\in D$, we introduce a \emph{chord-client} $c_e$, which is adjacent only to $s_e$.
We denote the constructed bipartite graph by $B$.

The order in which the clients are revealed is defined as follows.
First, the vertex-clients are revealed in an arbitrary order.
Next, the chord-clients are revealed in the order $c_{d_1}, c_{d_2}, \dots, c_{d_k}$.
\end{definition}

Observe that the constructed graph has a perfect matching:
We match every chord-client $c_e$ to~$s_e$.
It then only remains to match each vertex-client to an incident edge of the Hamiltonian cycle~$\hamcycle$, which is obviously possible (by choosing one of the two cyclic orientations of~$\hamcycle$ and matching each vertex to its outgoing edge).
Moreover, the number of servers/clients is~$|E(G)|$ as required for \cref{prop:matching_instance}.
Finally, note that the constructed graph is subcubic if $G$ is.

Before giving the adversary's strategy for constructing the matching, we note two facts that immediately follow from the definition of $B$.

The first observation has already been noted in previous work~\cite{BosekSAPtrees18, BosekSAPtrees22} and follows immediately from the following simple fact.
Let $c$ be a previously revealed client with unique neighbor $s$. In every matching that covers all revealed clients, $c$ is matched to $s$. Consequently, no augmenting path for a client revealed after~$c$ can pass through $s$ or $c$.

\begin{observation}
\label{obs:dead_vertices}
Let $c$ be a client whose unique neighbor is $s$.
Then, after $c$ is revealed, no augmenting path can pass through $s$ or $c$.
\end{observation}

The next lemma uses the straightforward link between paths in $G$ and paths in its incidence graph.
In the following, for a graph $G'$ and two sets of edges $E_1,E_2\subseteq E(G')$, we define
\[
\dist_{G'}(E_1,E_2):=\min\{\dist_{G'}(u,v): e_1\in E_1,\ e_2\in E_2,\ u\in e_1,\ v\in e_2\}.
\]
We also allow the arguments to be a single edge which we identify with the singleton containing it.

\begin{lemma}
\label{lem:paths_incidence_graph}
Let $i \in \{1,\dots, k\}$ and consider the time step $t$ of the process in which chord-client $c_{d_i}$ is revealed.
Let $M=M'_{t-1}$ be the matching presented by the adversary upon arrival of~$c_{d_i}$.
Let $D':=\{d_1,\dots, d_{i-1}\}$ denote the chords whose chord-clients have already been revealed, and let~$E'\subseteq E\setminus D'$ denote the set of edges $e'$ whose server $s_{e'}$ is unmatched in $M$.
Then the cost incurred by \sap in this time step is at least
\[
2 \dist_{G - D'}(d_i, E') + 1.
\]
\end{lemma}

\begin{proof}
Upon the arrival of $c_{d_i}$, the cost of \sap is the length of the augmenting path it chooses.
By definition, this augmenting path starts at $c_{d_i}$.
Since $c_{d_i}$ has the unique neighbor $s_{d_i}$, its next vertex is $s_{d_i}$.
The path ends at an unmatched server $s_{e'}$ for some $e'\in E'$.
Moreover, by \cref{obs:dead_vertices}, it cannot contain any server $s_d$ with $d\in D'$, nor any chord-client other than its initial vertex~$c_{d_i}$.
Consequently, after deleting the initial vertex $c_{d_i}$, we obtain a path $P$ in the incidence graph of~$G - D'$ that starts at $s_{d_i}$ and ends at $s_{e'}$ for some $e'\in E'$.

If $P$ has length zero, then $e'=d_i$, and the claim follows immediately.
Otherwise, write $P$ as
\[
P=(s_{e_1},c_{v_1},s_{e_2},c_{v_2},\dots,s_{e_\ell},c_{v_\ell},s_{e_{\ell+1}}),
\]
where $e_1=d_i$ and $e_{\ell+1}=e'$.
For every $j\in\{1,\dots,\ell-1\}$, the vertices $v_j$ and $v_{j+1}$ are adjacent in~$G - D'$ (via the edge $e_{j+1}\in E(G)\setminus D'$).
Moreover, $v_1$ is incident to $d_i$ and $v_\ell$ is incident to $e'$.
In particular,
\[
\ell\geq \ell-1 \geq \dist_{G - D'}(v_1,v_\ell)
\geq \dist_{G - D'}(d_i,e')
\geq \dist_{G - D'}(d_i,E').
\]
Thus the length of $P$ is at least
$
2\ell\geq 2\dist_{G - D'}(d_i,E').
$
Prefixing $P$ with $c_{d_i}$ yields that the length of the augmenting path is at least
$
2\dist_{G - D'}(d_i,E')+1,
$
as claimed.
\end{proof}

It remains to define suitable matchings for the adversary's strategy.
For this, the following lemma will be useful.

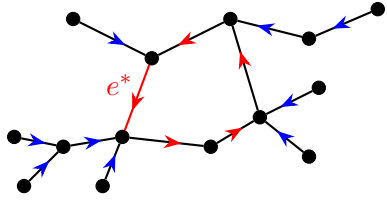
\begin{figure}
\centering
\begin{tikzpicture}[scale=1.3, every node/.style={circle, draw, fill, inner sep=0pt, minimum size=5pt}]
\node (A) at (0,0) {};
\node (B) at (0.9,-0.1) {};
\node (C) at (1.1,1.2) {};
\node (D) at (0.3,0.8) {};
\node (E) at (1.4,0.2) {};
\draw[thick] 
(A) edge[midarrow=red] (B)
(B) edge[midarrow=red] (E)
(E) edge[midarrow=red] (C)
(C) edge[midarrow=red] (D);
\draw[red, thick, midarrow=red] (D) to node [left, draw=none, fill=none, yshift=5pt] {$e^*$} (A);
\node (f) at (-0.6,-0.1) {};
\node (g) at (-0.2,-0.5) {};
\node (h) at (-1.1, 0) {};
\node (i) at (-1, -0.5) {};
\draw[thick] 
(f) edge[midarrow=blue] (A)
(g) edge[midarrow=blue] (A)
(h) edge[midarrow=blue] (f)
(i) edge[midarrow=blue] (f);
\node (j) at (1.9,1) {};
\node (k) at (2.6,1.3) {};
\draw[thick]
(j) edge[midarrow=blue] (C)
(k) edge[midarrow=blue] (j);
\node (l) at (-0.5,1.2) {};
\draw[thick]
(l) edge[midarrow=blue] (D);
\node (m) at (2,0.5) {};
\node (n) at (1.9,-0.2) {};
\draw[thick]
(m) edge[midarrow=blue] (E)
(n) edge[midarrow=blue] (E)
;
\end{tikzpicture}
\caption{The figure illustrates the construction given in \cref{lem:tree_covering}.
The graph is the union of a tree and one additional edge $e^*$. It has a unique cycle on which a cyclic orientation is chosen (indicated by red arrows).
All other edges are oriented towards the cycle (indicated by blue arrows).
In the resulting graph, every vertex has out-degree precisely 1.}
\label{fig:tree_covering}
\end{figure}

\begin{lemma}
\label{lem:tree_covering}
Let $T$ be a subgraph of $G$ that is a forest, and let $D'\subseteq D$.
Then there exists a matching in~$B$ that covers all servers $s_e$ with $e\in E(T)$, all vertex-clients, and all chord-clients $c_e$ with $e\in D'$.
\end{lemma}

\begin{proof}
First, extend $T$ to a spanning tree $T^*$ of $G$ and pick an edge $e^*\in E(G)\setminus T^*$, which exists since~$G$ is Hamiltonian.
Then $T^*\cup \{e^*\}$ contains a unique cycle $C^*$.
Choose a cyclic orientation of the edges of $C^*$, and orient all edges in $T^*\setminus C^*$ towards $C^*$ (cf.~\cref{fig:tree_covering}).
That is, if~$e=\{u,v\}\in T^*\setminus C^*$ and $\dist_{T^*}(u,C^*)=\dist_{T^*}(v,C^*)+1$, then we orient $e$ from $u$ to~$v$.
Observe that every vertex has out-degree precisely one in this orientation of $T^*\cup \{e^*\}$.
In the bipartite graph $B$, we now match every vertex-client $c_v$ to the server $s_e$, where $e$ is the unique edge outgoing from $v$ in the constructed orientation.
So far, the matching covers all vertex-clients and all servers~$s_e$ with $e\in T^*$, and hence in particular all servers~$s_e$ with $e\in T$.

It remains to cover the chord-clients $c_e$ with $e\in D'$.
We do this one by one.
Suppose that~$c_e$ is the next such chord-client and let $M$ be the matching constructed so far.
Since the final graph~$B$ containing all vertex- and chord-clients has a perfect matching, there exists an augmenting path from $c_e$ with respect to $M$.
After augmenting along this path, all clients and servers that were covered before remain covered, and~$c_e$ is covered as well.
Applying this argument to every $c_e$ with~$e\in D'$ yields a matching with the desired properties.
\end{proof}

Now, we have all the prerequisites at hand to prove \cref{prop:matching_instance}. We begin by restating it.

\propmatchinginstance*

\begin{proof}
First, recall that to prove the claimed lower bound for any algorithm, it suffices to argue for \sap.
The instance we construct is defined on the graph $B$, and the clients are revealed as described in \cref{def:matching_instance}.
During the arrival of the vertex-clients, the adversary may choose any maximum matching and we do not make any claims about the algorithm's cost in these steps.

For $i\in \{1,\dots,k\}$, it remains to construct the matching that the adversary presents upon revealing $c_{d_i}$.
Let $D':=\{d_1,\dots,d_{i-1}\}$ denote the chords whose chord-clients have already been revealed, let $G':=G-D'$, and let $g:=\girth(G')$.
Moreover, let
\[
T:=\{e\in E(G'):\dist_{G'}(d_i,e)\leq g/2-3\}
\]
be the set of edges in $G'$ at distance at most $g/2-3$ from $d_i$.

We first observe that $T$ is a tree.
Indeed, let $u$ be an endpoint of $d_i$.
Then every vertex incident to an edge in $T$ has distance at most $(g/2-3)+2=g/2-1$ from~$u$ in $T$.
If $T$ contained a cycle, then, considering a breadth-first search tree of $T$ rooted at~$u$, this cycle would contain an edge not belonging to the tree.
Together with the two tree paths from $u$ to the endpoints of this edge, this gives a cycle in $T\subseteq G'$ of length at most
$
2(g/2-1)+1<g,
$
contradicting the girth of~$G'$ being $g$.
By \cref{lem:tree_covering}, applied to $T$ and $D'$, there exists a matching $M$ that covers all servers $s_e$ with $e\in T$, all vertex-clients, and all chord-clients $c_e$ with $e\in D'$.
We let the adversary present this matching upon revealing the client $c_{d_i}$.

Let $E'\subseteq E(G')$ be the set of edges $e$ whose server $s_e$ is not covered by $M$.
Since $M$ covers all servers corresponding to edges in $T$, every edge in $E'$ has distance more than $g/2-3$ from~$d_i$ in~$G'$.
By \cref{lem:paths_incidence_graph}, the cost incurred by \sap in this step is therefore at least
\[
2\dist_{G'}(d_i,E')+1
>
2(g/2-3)+1
=
g-5
=\girth(G')-5.
\]
Summing over all $i\in\{1,\dots,k\}$ yields the statement of the proposition.
\end{proof}

\subsection{Constructing layered high-girth chord graphs}
\label{sec:construct_G}

In this subsection, we construct the family of graphs needed for our lower bound, i.e., we prove \cref{prop:construction_G}.
We begin by proving two lemmas that help us in constructing graphs of high girth.

The next lemma captures a simple technique to increase the girth of a graph. Here, subdividing an edge $s-1$ times means replacing it by a path of length $s$.

\begin{lemma}
\label{lem:girth_subdivide}
Let $G$ be a graph with Hamiltonian cycle $\hamcycle$ such that its chords $D=E\setminus \hamcycle$ are a matching.
Let $s\geq 1$ be an integer, and let $G'$ be obtained from $G$ by subdividing each edge of $\hamcycle$ exactly $s-1$ times.
Then 
\begin{equation*}
    \girth(G')\geq \frac{s+1}{2} \girth(G).
\end{equation*}
\end{lemma}

\begin{proof}
Let $Q$ be a cycle in $G'$ of length $\girth(G')$.
By contracting the subdivided edges of $Q$, we obtain a cycle $Q^*$ in $G$, and let $g$ denote its length.
Since $D$ is a matching, the cycle $Q^*$ contains at most~$g/2$ edges from $D$.
Hence $Q^*$ contains at least $g/2$ edges from $\hamcycle$.
Each edge of $\hamcycle$ has been replaced in $G'$ by a path of length $s$, and the edges in $D$ remain unchanged.
Therefore,
\begin{align*}
    \girth(G')&= |Q| 
    = s \cdot |Q^* \cap H| + |Q^* \cap D|
    =|Q^*|+ (s-1)\cdot |Q^* \cap H|\\
    &\geq g+ (s-1) \cdot \frac{g}{2} = \frac{s+1}{2} g
    \geq \frac{s+1}{2} \girth(G).\qedhere
\end{align*}
\end{proof}

Next, we prove a property that allows us to add many edges to a graph while maintaining logarithmic girth.

\begin{lemma}
\label{lem:add_matching}
Let $G$ be a graph on $n\geq 16$ vertices such that every vertex has degree 2 or 3, and assume that $\girth(G)\geq  \log n /2$.
Then one can add a matching between the degree-2 vertices of $G$ such that the resulting graph $G'$ still satisfies $\girth(G')\geq \log n /2$ and has at most $\sqrt{n}$ vertices of degree~2.
\end{lemma}

\begin{proof}
First, observe the following simple fact. Given a graph of girth at least $g$ and two vertices~$u$ and~$v$ at distance at least $g-1$, adding the edge $\{u,v\}$ maintains that the girth is at least $g$. Indeed, any cycle not containing the edge $\{u,v\}$ existed before and thus has length at least $g$. Any new cycle consists of the new edge $\{u,v\}$ and a $u$-$v$-path in the previous graph, so its length is at least~$(g-1)+1=g$. 

We make use of this fact as follows.
Initialize $G'$ to be $G$.
At each step, let $X$ denote the set of vertices that have degree $2$ in the current graph $G'$.
If possible, choose two non-adjacent vertices~$u,v\in X$ with $\dist_{G'}(u,v)\geq \log n/2 -1$, and add the edge $\{u,v\}$ to $G'$.
After this step, both~$u$ and $v$ have degree $3$.
By the observation above, the girth remains at least $\log n/2$ throughout the process.
It remains to show that the process can only stop once $|X|\leq \sqrt n$.

For this, observe that the graph $G'$ is subcubic throughout the process.
Therefore, for any vertex $v\in X$ and $i\geq 1$, the number of vertices at distance exactly $i$ from $v$ is at most $2^i$ (the size of the $i$-th layer of a complete binary tree).
Hence, the number of vertices at distance at most~$\log n/2 -1$ from $v$ is at most
\begin{equation*}
1+ \sum\limits_{i=1}^{\lfloor \log n/2 \rfloor -1} 2^i \leq  2^{\lfloor \log n/2 \rfloor} \leq 2^{\log n /2} = \sqrt{n}.
\end{equation*}
Thus, if $|X|>\sqrt n$, then for every $v\in X$ there exists a vertex $u\in X$ with $\dist_{G'}(u,v)>  \log n/2 -1$.
In particular, since $n\geq 16$, we have $\dist_{G'}(u,v)>1$ so $u$ and $v$ are non-adjacent.
In this case the process can continue, i.e., we can add a new edge between $u$ and $v$.
Consequently, the process stops only once $|X|\leq \sqrt n$.
This completes the proof of the lemma.
\end{proof}

Now, we have all the prerequisites at hand to prove \cref{prop:construction_G}.
We begin by restating it.

\propconstructiong*

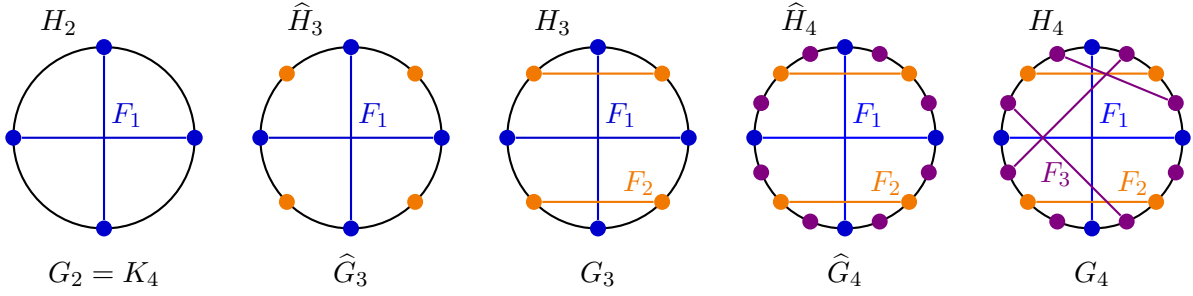
\begin{figure}
\centering
\begin{tikzpicture}[scale=0.6]
    \def\R{2}    \colorlet{gtwo}{blue!80!black}
    \colorlet{gthree}{orange!95!black}
    \colorlet{gfour}{violet}
        \tikzset{
        gtwovertex/.style={circle, fill=gtwo, inner sep=0pt, minimum size=6pt},
        gthreevertex/.style={circle, fill=gthree, inner sep=0pt, minimum size=6pt},
        gfourvertex/.style={circle, fill=gfour, inner sep=0pt, minimum size=6pt}
    }
    \draw[thick] (0,0) circle (\R);
    \foreach \i in {1,...,4} {
        \node[gtwovertex] (B\i) at ({90*(\i-1)}:\R) {};
    }
    \draw[thick, gtwo] (B1) to (B3);
    \draw[thick, gtwo] (B2) to (B4);
    \node at (-1,2.05) [anchor=south] {$H_2$};
    \node at (0.5,0.5) [gtwo] {$F_1$};
    \node at (0,-3.5) [anchor=south] {$G_2=K_4$};
    \node at (0,3) {};
\end{tikzpicture}
\hspace{4mm}
\begin{tikzpicture}[scale=0.6]
    \def\R{2}    \colorlet{gtwo}{blue!80!black}
    \colorlet{gthree}{orange!95!black}
    \colorlet{gfour}{violet}
        \tikzset{
        gtwovertex/.style={circle, fill=gtwo, inner sep=0pt, minimum size=6pt},
        gthreevertex/.style={circle, fill=gthree, inner sep=0pt, minimum size=6pt},
        gfourvertex/.style={circle, fill=gfour, inner sep=0pt, minimum size=6pt}
    }
    \draw[thick] (0,0) circle (\R);
    \foreach \i in {1,...,4} {
        \node[gtwovertex] (B\i) at ({90*(\i-1)}:\R) {};
    }
    \foreach \i in {1,...,4} {
        \node[gthreevertex] (G\i) at ({45 + 90*(\i-1)}:\R) {};
    }
    \draw[thick, gtwo] (B1) to (B3);
    \draw[thick, gtwo] (B2) to (B4);
    \node at (-1,2.05) [anchor=south] {$\widehat H_3$};
    \node at (0.5,0.5) [gtwo] {$F_1$};
    \node at (0,-3.5) [anchor=south] {$\widehat G_3$};
    \node at (0,3) {};
\end{tikzpicture}
\hspace{4mm}
\begin{tikzpicture}[scale=0.6]
    \def\R{2}    \colorlet{gtwo}{blue!80!black}
    \colorlet{gthree}{orange!95!black}
    \colorlet{gfour}{violet}
        \tikzset{
        gtwovertex/.style={circle, fill=gtwo, inner sep=0pt, minimum size=6pt},
        gthreevertex/.style={circle, fill=gthree, inner sep=0pt, minimum size=6pt},
        gfourvertex/.style={circle, fill=gfour, inner sep=0pt, minimum size=6pt}
    }
    \draw[thick] (0,0) circle (\R);
    \foreach \i in {1,...,4} {
        \node[gtwovertex] (B\i) at ({90*(\i-1)}:\R) {};
    }
    \foreach \i in {1,...,4} {
        \node[gthreevertex] (G\i) at ({45 + 90*(\i-1)}:\R) {};
    }
    \draw[thick, gtwo] (B1) to (B3);
    \draw[thick, gtwo] (B2) to (B4);
    \node at (-1,2.05) [anchor=south] {$H_3$};
    \node at (0.5,0.5) [gtwo] {$F_1$};
    \draw[thick,gthree] (G1) to (G2);
    \draw[thick,gthree] (G3) to (G4);
    \node at (0.9,-1) [gthree] {$F_2$};
    \node at (0,-3.5) [anchor=south] {$G_3$};
    \node at (0,3) {};
\end{tikzpicture}
\hspace{4mm}
\begin{tikzpicture}[scale=0.6]
    \def\R{2}    \colorlet{gtwo}{blue!80!black}
    \colorlet{gthree}{orange!95!black}
    \colorlet{gfour}{violet}
        \tikzset{
        gtwovertex/.style={circle, fill=gtwo, inner sep=0pt, minimum size=6pt},
        gthreevertex/.style={circle, fill=gthree, inner sep=0pt, minimum size=6pt},
        gfourvertex/.style={circle, fill=gfour, inner sep=0pt, minimum size=6pt}
    }
    \draw[thick] (0,0) circle (\R);
    \foreach \i in {1,...,4} {
        \node[gtwovertex] (B\i) at ({90*(\i-1)}:\R) {};
    }
    \foreach \i in {1,...,4} {
        \node[gthreevertex] (G\i) at ({45 + 90*(\i-1)}:\R) {};
    }
    \foreach \i in {1,...,8} {
        \node[gfourvertex] (X\i) at ({22.5 + 45*(\i-1)}:\R) {};
    }
    \draw[thick, blue] (B1) to (B3);
    \draw[thick, blue] (B2) to (B4);
    \node at (-1,2.05) [anchor=south] {$\widehat H_4$};
    \node at (0.5,0.5) [blue] {$F_1$};
    \draw[thick,gthree] (G1) to (G2);
    \draw[thick,gthree] (G3) to (G4);
    \node at (0.9,-1) [gthree] {$F_2$};
    \node at (0,-3.5) [anchor=south] {$\widehat G_4$};
    \node at (0,3) {};
\end{tikzpicture}
\hspace{4mm}
\begin{tikzpicture}[scale=0.6]
    \def\R{2}    \colorlet{gtwo}{blue!80!black}
    \colorlet{gthree}{orange!95!black}
    \colorlet{gfour}{violet}
        \tikzset{
        gtwovertex/.style={circle, fill=gtwo, inner sep=0pt, minimum size=6pt},
        gthreevertex/.style={circle, fill=gthree, inner sep=0pt, minimum size=6pt},
        gfourvertex/.style={circle, fill=gfour, inner sep=0pt, minimum size=6pt}
    }
    \draw[thick] (0,0) circle (\R);
    \foreach \i in {1,...,4} {
        \node[gtwovertex] (B\i) at ({90*(\i-1)}:\R) {};
    }
    \foreach \i in {1,...,4} {
        \node[gthreevertex] (G\i) at ({45 + 90*(\i-1)}:\R) {};
    }
    \foreach \i in {1,...,8} {
        \node[gfourvertex] (X\i) at ({22.5 + 45*(\i-1)}:\R) {};
    }
    \draw[thick, blue] (B1) to (B3);
    \draw[thick, blue] (B2) to (B4);
    \node at (-1,2.05) [anchor=south] {$H_4$};
    \node at (0.5,0.5) [blue] {$F_1$};
    \draw[thick,gthree] (G1) to (G2);
    \draw[thick,gthree] (G3) to (G4);
    \node at (0.9,-1) [gthree] {$F_2$};
    \draw[thick,gfour] (X1) to (X3);
    \draw[thick,gfour] (X5) to (X2);
    \draw[thick,gfour] (X4) to (X7);
    \node at (-0.8,-0.8) [gfour] {$F_3$};
    \node at (0,-3.5) [anchor=south] {$G_4$};
    \node at (0,3) {};
\end{tikzpicture}
\caption{The family of graphs constructed inductively in \cref{prop:construction_G}.
In each subfigure, the black edges belong to the Hamiltonian cycle $H_L$, respectively $\widehat{H_L}$, and the colored edges belong to~$F_1$, respectively $F_2$ or $F_3$. The blue vertices belong to $G_2$, the orange vertices are added during the construction of $G_3$ from $G_2$, and the purple vertices during the construction of $G_4$ from $G_3$.}
\label{fig:graph_construction}
\end{figure}

\begin{proof}
We prove the statement for $\alpha=\nicefrac{1}{4}$ and $\beta=\nicefrac{1}{8}$.

We build the family of graphs inductively as follows.
First, we construct graphs~$G_L$ on~$2^L$ vertices, each with a Hamiltonian cycle $H_L$, such that $E(G_L)\setminus \hamcycle_L$ is a matching and
$\girth(G_L)\geq  L/2$.
The desired partition of the chords will be defined afterwards.

For the base case, let $G_2:=K_4$ (cf.~\cref{fig:graph_construction}, left).
Choose any Hamiltonian cycle $H_2$ in~$K_4$.
Then $|V(G_2)|=2^2$, the chord set $E(G_2)\setminus \hamcycle_2$ is a matching of size 2, and $\girth(G_2)=3$.

Now suppose that $G_L$ with cycle $H_L$ has already been constructed.
We define an auxiliary graph~$\widehat G_{L+1}$ by subdividing each edge of $H_L$ exactly once (cf.~\cref{fig:graph_construction}).
Let $\widehat H_{L+1}$ be the Hamiltonian cycle obtained from~$H_L$ in this way.
Then $\widehat G_{L+1}$ has $2^{L+1}$ vertices and is subcubic.
Also, it is Hamiltonian so every vertex has degree at least 2.
Moreover, all $2^L$ newly introduced vertices have degree $2$.
By induction hypothesis, we have $\girth(G_L)\geq  L/2$ and, by \cref{lem:girth_subdivide} (with $s=2$), the girth of the subdivided graph $\widehat G_{L+1}$ is at least $\nicefrac{3}{2} \cdot L/2 \geq  (L+1)/2$ (where we used~$L\geq 2$).
In case~$L+1\geq 4$, by \cref{lem:add_matching}, we can add a matching between the degree-2 vertices maintaining that the girth is at least $ (L+1)/2$ such that at most $\sqrt{2^{L+1}}$ vertices of degree 2 remain.
In the special case $L+1=3$, we simply add any matching of two edges between the degree-2 vertices (cf.~$G_3$ in \cref{fig:graph_construction}). The resulting graph in this case has no degree-2 vertices and trivially has girth at least $3\geq (L+1)/2$.

In both cases, the total number of vertices that are covered by the new matching is at least
\begin{equation*}
    2^L-\lfloor \sqrt{2^{L+1}} \rfloor
    \geq 2^L-2^{L-1}= \frac{1}{2}2^L.
\end{equation*}
Here, for \(L\ge 3\) we use \(\sqrt{2^{L+1}}=2^{(L+1)/2}\le 2^{L-1}\), and the case \(L=2\) is checked separately (where we have $\lfloor \sqrt{2^3} \rfloor = 2 = 2^1$).
It follows that the number of edges added in this matching is at least half of the newly covered vertices, i.e., $2^L/4$.
Also, since the matching is added only between degree-2 vertices, the set of all chords remains a matching.

So far, we have constructed, for every $L\geq 2$, a graph $G_L$ with Hamiltonian cycle $H_L$, proved that $\girth(G_L)\geq L/2$, and we have noted that the number of edges added during the construction of~$G_{i+1}$ from $G_i$ is at least $2^i/4$.
For $i\geq 2$, we define the set $F_i$ to be precisely these edges, where the set $F_1$ is simply defined as the set of the two chords of $G_2$.
In particular, we have $|F_i|\geq 2^i/4$ as required.

It only remains to verify the girth condition \ref{cond:girth}.
For this, observe that, for every $L\geq i+1$, the cycle $H_L$ is obtained from $H_{i+1}$ by subdividing each edge precisely $2^{L-(i+1)}$-1 times (cf.~\cref{fig:graph_construction}).
With \cref{lem:girth_subdivide}, we obtain
\begin{align*}
\girth \Bigl(H_L \cup {\textstyle \bigcup_{l=1}^i} F_l \Bigl) 
\overset{\text{Lem.}\ref{lem:girth_subdivide}}&{\geq}
\frac{2^{L-(i+1)}+1}{2} \girth \Bigl(H_{i+1} \cup {\textstyle \bigcup_{l=1}^i} F_l \Bigl)
= \frac{2^{L-i-1}+1}{2} \girth(G_{i+1})\\
&\geq \frac{2^{L-i-1}+1}{2} \cdot \frac{i+1}{2}
\geq \frac{2^{L-i-1}}{2} \cdot \frac{i}{2}
=\frac{1}{8} \cdot i \cdot 2^{L-i}.\qedhere
\end{align*}
\end{proof}

As established in the beginning of this section, Propositions~\ref{prop:matching_instance} and \ref{prop:construction_G} together imply \cref{thm:lowerbound}, thus this completes the proof of our main result.

\section{Upper bound for expander graphs}
\label{sec:expander}

In this section, we prove our upper bound for expander graphs in the malicious setting (cf.~\cref{thm:expanders}).

First, we give the definition of expanders that we work with here.
For a~graph $G = (V, E)$ and a~nonempty set $W \subseteq V$ of vertices with $|W| \leq \frac12 |V|$, define
\[
    h_G(W) \coloneqq \frac{|\partial_G(W)|}{|W|}, \qquad\text{where}\ \partial_G(W) \coloneqq \{\{u,v\} \in E \,\mid\, u \in W, v \notin W\}.
\]
In other words, $h_G(W)$ is the ratio of the number of edges crossing the cut $(W, V \setminus W)$ to the size of $W$.
Then the \emph{edge expansion} of a~graph $G$ is
\[
    h(G) \coloneqq \min \{ h_G(W) \,\mid\, W \subseteq V,\, 1 \leq |W| \leq \tfrac{1}{2} |V| \}.
\]
In particular, we say that $G$ is an \emph{$h$-edge expander} for some $h > 0$ if $h(G) \geq h$.

We lift this definition to bipartite graphs as follows. Let $B=\tuple{S,C,E_B}$ be a bipartite graph, whose vertex set $V(B)=S \sqcup C$ is a disjoint union of the set of servers $S$ and the set of clients~$C$. 
For a~real number $h > 0$ and an~integer $d \geq 1$, we say that $B$ \emph{contains a~$d$-regular $h$-edge-expanding spanning subgraph} if it contains as a~subgraph a~$d$-regular bipartite graph $\widehat{B} = \tuple{S,C, \widehat{E}}$ with~$\widehat{E} \subseteq E_B$ that is an~$h$-edge expander.
Note that this definition requires that the vertex set of the expander $\widehat{B}$ is precisely $V(B)$, and that the expander is a~regular graph.
Also, observe that such a graph contains a perfect matching (because every regular bipartite graph does).

We now move on to the proof of the $\bigO(n \log n)$ recourse upper bound in the case where the client-server graph contains an~expander graph, i.e., the proof of \cref{thm:expanders}.
We begin by restating it.

\expanderupperbound*

The proof of \cref{thm:expanders-formal} is akin to the standard argument that constant-degree edge expanders have logarithmic diameter; in fact, we will show that on each iteration, \sap{} locates an~augmenting path of logarithmic length.
However, particular care must be taken to address the fact that on each iteration, \sap{} must locate an~augmenting path from the most recently introduced client to an~unmatched server (and not merely an~arbitrary path between these vertices).
In particular, the $d$-regularity of the spanning subgraph will play a~crucial role in our argument, as shown below.

\begin{proof}[Proof of \cref{thm:expanders-formal}]
    Without loss of generality, we can always replace $h$ with $\min(h, \frac12)$, and thus we can assume that $h \in (0, 1)$. It follows that $h < d$.
    As promised, we will show that on each iteration, \sap{} will locate a~logarithmic-length augmenting path.
    Let $B=\tuple{S,C,E_B}$ be a bipartite graph with $|S|=|C|=n$ that contains a $d$-regular $h$-edge expander $\widehat{B}=\tuple{S,C,\widehat{E}}$ as in the theorem statement.
    Let $C' \subseteq C$ be the set of already introduced clients, $B' = B[C' \sqcup S]$ be the corresponding portion of $B$, $c' \in C'$ be the most recently introduced (and yet unmatched) client, and $M$ be the matching in $B'$ covering $C' \setminus \{c'\}$.
    Throughout the proof, whenever we speak of alternating paths, we mean an~alternating path in $B'$ with respect to the presented matching~$M$.
    Define $C_i\subseteq C'$, $i \geq 0$, as the set of clients reachable from $c'$ via an~alternating path of length at most~$2i$ in $B'$, and~$S_i$ as the set of servers reachable from $c'$ via an~alternating path of length at most $2i + 1$ in $B'$.

\begin{figure}
\centering

\begin{tikzpicture}[
  x=1.35cm,y=1.25cm,
  dot/.style={circle,fill=black,inner sep=0pt,minimum size=5pt},
  cprime/.style={circle,fill=special!70, draw=black,
                 inner sep=0pt,minimum size=6pt},
  edge/.style={unmatched,line width=1pt,line cap=round},
  match/.style={matching,line width=2pt,line cap=round},
  start/.style={edge,dashed,dash pattern=on 4pt off 3pt},
  reg/.style={line width=.9pt,rounded corners=8pt},
  lab/.style={font=\small}
]

\colorlet{colA}{blue!70!black}
\colorlet{colB}{orange!80!black}
\colorlet{colC}{purple!70!black}
\colorlet{colnew}{yellow!70!black}
\colorlet{special}{magenta!55}
\colorlet{matching}{black}
\colorlet{unmatched}{black!70}

\coordinate (c) at (0,0);
\foreach \i [evaluate=\i as \x using \i-1] in {1,...,9}
  \coordinate (u\i) at (\x,2);
\foreach \i in {1,...,8}
  \coordinate (v\i) at (\i,0);

\newcommand{\reg}[4]{  \filldraw[reg,fill=#1!#2,draw=#1]
    ($(#3)+(-.35,-.38)$) rectangle ($(#4)+(.35,.38)$);}

\begin{scope}[on background layer]
  \reg{colC}{15}{c}{v4}
  \reg{colB}{20}{c}{v3}
  \reg{colA}{20}{c}{v2}
  \reg{colnew}{20}{c}{c}
  \reg{colC}{15}{u1}{u7}
  \reg{colB}{20}{u1}{u4}
  \reg{colA}{20}{u1}{u3}
  \reg{colnew}{20}{u1}{u2}
\end{scope}

\node[lab,text=colnew!60!black] at ($(c)+(0,-.65)$) {$C_0$};
\node[lab] at ($(c)!.5!(v2)+(0,-.65)$) {$\subseteq$};
\node[lab,text=colA] at ($(v2)+(0,-.65)$) {$C_1$};
\node[lab] at ($(v2)!.5!(v3)+(0,-.65)$) {$\subseteq$};
\node[lab,text=colB] at ($(v3)+(0,-.65)$) {$C_2$};
\node[lab] at ($(v3)!.5!(v4)+(0,-.65)$) {$\subseteq$};
\node[lab,text=colC] at ($(v4)+(0,-.65)$) {$C_3$};

\node[lab,text=colnew!60!black] at ($(u1)+(0,.65)$) {$S_0$};
\node[lab] at ($(u1)!.5!(u3)+(0,.65)$) {$\subseteq$};
\node[lab,text=colA] at ($(u3)+(0,.65)$) {$S_1$};
\node[lab] at ($(u3)!.5!(u4)+(0,.65)$) {$\subseteq$};
\node[lab,text=colB] at ($(u4)+(0,.65)$) {$S_2$};
\node[lab] at ($(u4)!.5!(u5)+(0,.65)$) {$\subseteq$};
\node[lab,text=colC,anchor=west] at ($(u5)+(-.15,.65)$)
  {$S_3 = N_{B'}(C_3)$};

\draw[start] (c)--(u1) (c)--(u2);

\foreach \i in {1,...,5}
  \draw[match] (u\i)--(v\i);
\draw[match] (u8)--(v6);
\draw[edge] (u9)--(v7);
\draw[match] (u9)--(v8);
\draw[match] (u7)--(v7);

\draw[edge]
  (v1)--(u3)
  (v2)--(u3)
  (v3)--(u4)
  (u4)--(v5)
  (v4)--(u5)
  (v4)--(u6)
  (v4)--(u7)
  (v7)--(u8)
  (v7)--(u7)
  (v7)--(u9)
  (v8)--(u9);

\node[cprime] at (c) {};
\foreach \i in {1,...,9} \node[dot] at (u\i) {};
\foreach \i in {1,...,8} \node[dot] at (v\i) {};

\node[font=\large] at ($(u1)+(-.75,0)$) {$S$};
\node[font=\large] at ($(c)+(-.75,0)$) {$C'$};
\node[lab] at ($(c)+(0,-.20)$) {$c'$};

\end{tikzpicture}

\caption{The figure illustrates a graph $B' = B[C' \sqcup S]$ at the moment of introducing client $c' \in C'$ together with its edges to $S$, just before matching it.
The presented matching $M$ covering $C'\setminus\{c'\}$ is depicted by thick edges. 
Sets $C_i$, $S_i$ depict clients and servers reachable by alternating paths from $c'$ of length at most $2i$ and $2i+1$, respectively.}
\label{fig:layersets}
\end{figure}
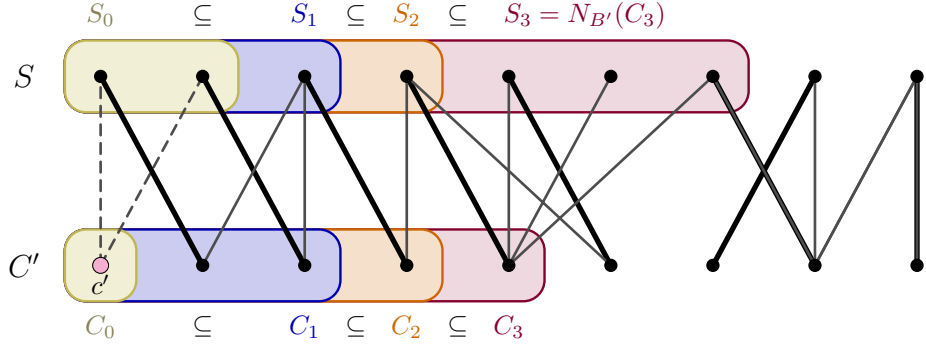

    Observe first that $S_i = N_{B'}(C_i)$ (cf.~\cref{fig:layersets}): on the one hand, if $s \in S_i$, then consider the client $c$ that is the penultimate vertex on the alternating path from $c'$ to $s$.
    Then $s \in N_{B'}(c)$ and $c \in C_i$.
    Conversely, let $s \in N_{B'}(C_i)$ and $c \in C_i$ be the neighbor of $s$ in $B'$.
    If $\{c,s\} \notin M$, consider the alternating path from $c'$ to $c$ of length at most $2i$.
    If $s$ already belongs to this path, then trivially $s \in S_i$.
    Otherwise, form an~alternating path from $c'$ to $s$ of length at most $2i+1$ by appending the edge $\{c,s\}$; this also proves that $s \in S_i$.
    If $\{c,s\} \in M$, then the alternating path from $c'$ to $c$ (of length at most $2i$) already has $\{c,s\}$ as the final edge.
    In particular, $i \geq 1$ and $s$ is reachable from $c'$ via an~alternating path of length at most $2i - 1$, implying $s \in S_{i-1} \subseteq S_i$.

    Next, if $S_i$ only contains matched servers, then $|C_{i+1}| = |S_i| + 1$.
    This holds because $C_{i+1}$ contains exactly all clients matched to $S_i$ and the additional unmatched source client $c'$.

    Suppose $S_\ell$ only contains matched servers, and let $k_i = |C_i|$ for $i = 0, \ldots, \ell + 1$.
    Our aim is to show that the values $k_i$ increase exponentially quickly on an~initial prefix, and then -- on the possible remaining suffix -- the values $n - k_i$ decrease exponentially quickly.
    This will in turn imply that $\ell \leq \bigO(\log n)$.
        In particular, the first index \(L\) for which \(S_L\) contains an unmatched server also satisfies \(L = \bigO(\log n)\), and hence there is an augmenting path of length at most \(2L+1 = \bigO(\log n)\), thus completing the proof.

    For $i \leq \ell$, let $W_i = C_i \cup S_i$.
    Then $W_0 \subseteq W_1 \subseteq \ldots \subseteq W_\ell$, and $|W_i| = k_i + k_{i+1} - 1$ for $i \leq \ell$.
    Now:
    \begin{itemize}
        \item If $|W_i| \leq n$, then $|W_i| \leq \frac12 |V(\widehat{B})|$ and by the fact that $\widehat{B}$ is an~$h$-edge expander, $|\partial_{\widehat{B}}(W_i)| \geq h |W_i|$.
            Also, since $S_i = N_{B'}(C_i)$, then also $S_i \supseteq N_{\widehat{B}}(C_i)$ and therefore every edge of $\partial_{\widehat{B}}(W_i)$ has one endpoint in $S_i$ and the other endpoint in $C \setminus C_i$.
            By $d$-regularity of $\widehat{B}$, the total degree in $\widehat{B}$ of all vertices in $C_i$ is $d|C_i| = dk_i$; hence $\widehat{B}[W_i]$ has exactly $dk_i$ edges.
            By the same token, the total degree in $\widehat{B}$ of all vertices in $S_i$ is $d|S_i| = d(k_{i+1} - 1)$.
            Therefore, $|\partial_{\widehat{B}}(W_i)| = d|S_i| - d|C_i| = d(k_{i+1} - k_i - 1)$.
            We conclude that
            \[
                d(k_{i+1} - k_i - 1) \geq h(k_i + k_{i+1} - 1),
            \]
            or equivalently,
            \[
                k_{i+1} \geq k_i \cdot \frac{d + h}{d - h} + 1.
            \]
        
        \item If $|W_i| > n$, then let $\bar{C}_i \coloneqq C \setminus C_i$, $\bar{S}_i \coloneqq S \setminus S_i$, and $\bar{W}_i \coloneqq \bar{C}_i \cup \bar{S}_i$ so that $|\bar{W}_i| = 2n - |W_i| \leq \frac12 |V(\widehat{B})|$.
            Then $|\partial_{\widehat{B}}(\bar{W}_i)| \geq h|\bar{W}_i|$.
            Also, $\partial_{\widehat{B}}(\bar{W}_i) = \partial_{\widehat{B}}(W_i)$, so by the same argument as above, $|\partial_{\widehat{B}}(\bar{W}_i)| = d(k_{i+1} - k_i - 1)$.
            Therefore,
            \[
                d(k_{i+1} - k_i - 1) \geq h(2n - k_i - k_{i+1} + 1),
            \]
            or equivalently,
            \[
                n - k_{i+1} \leq \frac{d - h}{d + h}(n - k_i) - 1.
            \]
    \end{itemize}
    In particular, the integers $k_0, \ldots, k_{\ell + 1}$ are strictly increasing, and obviously in the range $[1, n]$ as~$k_i=|C_i|$.

    Now let $\theta \coloneqq \frac{d + h}{d - h} > 1$.
    Also let $p \in \{-1, 0, \ldots, \ell\}$ be an~integer such that $|W_i| \leq n$ for all $i = 0, \ldots, p$, and $|W_i| > n$ for all $i = p + 1, \ldots, \ell$.
    Then, for $i = 0, \ldots, p$, we have $k_{i+1} > \theta k_i$.
    Therefore, $p \leq \log_\theta n$.
    Similarly, letting $\bar{k}_i = n - k_i$, we have $1 \leq \bar{k}_i \leq n$ for $i = p + 1, \ldots, \ell$, and $\bar{k}_{i+1} < \bar{k}_i / \theta$ for $i = p + 1, \ldots, \ell - 1$.
    It follows that $\ell - (p + 1) \leq \log_\theta n$.
    We conclude that $\ell \leq 1 + 2\log_\theta n$.
    Thus, let $L$ be the first index for which $S_L$ contains an~unmatched server.
    By our discussion above, we have $L \leq 2 + 2\log_\theta n$.
    It follows that there exists an~augmenting path in $B'$ from $c'$ of length at most $2L + 1 \leq \bigO_\theta(\log n) = \bigO_{d, h}(\log n)$.
    Since \sap{} always chooses the shortest possible augmenting path, the recourse of the algorithm is $\bigO_{d, h}(\log n)$ at every step.

    Repeating the argument above for each of $n$ arrivals, we bound the total recourse of \sap{} by~$\bigO_{d, h}(n \log n)$.
\end{proof}

\paragraph{Acknowledgements.} The authors thank Aditi Dudeja, Jacob Holm and Eva Rotenberg for insightful discussions preceding the development of this result.

 \bibliographystyle{plain}
 \bibliography{ArXiv-Version-1/MaliciousOnlineBipartiteMatching-References}

\end{document}